\documentclass[11pt]{article}
\usepackage[T1]{fontenc}
\usepackage{lmodern}
\usepackage[margin=1in]{geometry}
\usepackage{amsmath,amssymb,amsthm,mathtools}
\usepackage{booktabs,enumitem,microtype}
\usepackage[hidelinks]{hyperref}
\numberwithin{equation}{section}
\newtheorem{theorem}{Theorem}[section]
\newtheorem{proposition}[theorem]{Proposition}
\newtheorem{corollary}[theorem]{Corollary}
\newtheorem{lemma}[theorem]{Lemma}
\theoremstyle{definition}
\newtheorem{definition}[theorem]{Definition}
\theoremstyle{remark}

\newcommand{\C}{\mathbb C}

\newcommand{\Id}{I}
\newcommand{\Tr}{\operatorname{Tr}}
\newcommand{\spec}{\operatorname{spec}}

\newcommand{\Qq}{\mathcal Q_{\mathrm q}}
\newcommand{\Qf}{\mathcal Q_{\mathrm{frPVM}}}
\newcommand{\ket}[1]{\lvert #1\rangle}
\newcommand{\bra}[1]{\langle #1\rvert}
\newcommand{\norm}[1]{\left\lVert #1\right\rVert}

\newcommand{\Szero}{\mathsf S_0}
\newcommand{\Strine}{\mathsf S_1}
\newcommand{\Sref}{\mathsf S_H}
\setlist[enumerate]{label=\textup{(\roman*)},leftmargin=*,itemsep=3pt}
\allowdisplaybreaks[1]
\title{A Separation between Full-Rank PVM\protect\\
and Assumption-free Self-Testing}
\author{
  Ranyiliu Chen\\
  {\small Quantum Science Center of Guangdong--Hong Kong--Macao Greater Bay Area}\\
  {\small\href{mailto:chenranyiliu@quantumsc.cn}
    {\texttt{chenranyiliu@quantumsc.cn}}}
}
\date{}

\begin{document}
\maketitle

\begin{abstract}
  We construct a nonlocal game that self-tests a maximally entangled qubit strategy among pure full-Schmidt-rank projective strategies, but admits an inequivalent optimum using one nonprojective measurement. This resolves a conjecture of Baptista et al. on imposing full rank and projectivity simultaneously. The construction combines CHSH with an auxiliary game $G$ that forces deterministic answers under these assumptions and admits a trine POVM optimum without them. We classify the optimal correlations of $G$ as a line segment parametrized by the tracial states of $\C\oplus M_2(\C)$. Projectivity on the state support removes the matrix summand, whereas projective dilations with nonzero nonabort probability have a nontracial local state whose zero left ideal is not two-sided. Finally, a family in local dimension six shows that the full-rank PVM self-test is not robust.
\end{abstract}

\section*{\normalsize{Statement on the Use of Artificial Intelligence}}

The author(s) formulated the research question. Several frontier LLMs was used under human direction to generate exploratory proofs, prepare and revise the manuscript. The author(s) then refined the logical flow of the proof, guided the revision, and verified the statements. The author(s) take full responsibility for the content and correctness of this article.

\section{Introduction}

Self-testing provides a way to certify entangled states and measurement effects from the value of a nonlocal game, up to local isometries and an auxiliary state. Its conclusion concerns every optimal strategy in the class under consideration. We study how that conclusion changes when the shared state is required to be pure and of full Schmidt rank, and the measurements are required to be projective.

Baptista et al.~\cite{Baptista} give general results for removing such assumptions from self-testing theorems. For a full-rank projective reference strategy, one may separately start from pure projective strategies or from pure full-rank strategies. Requiring both full rank and projectivity at once presents a different problem. Naimark dilation replaces positive operator-valued measurements by projective measurements while preserving correlations, but can introduce local dimensions outside the support of the state. Compression to the Schmidt supports restores full rank, but a compression of a projection need not be a projection. The two standard reductions consequently do not supply a realization that satisfies both assumptions. Baptista et al. conjecture that this obstruction can persist in an exact self-testing problem even when the reference strategy itself satisfies both assumptions.

We prove their conjecture with an explicit game $H$ having six questions for Alice, five for Bob, and a maximally entangled reference strategy. The verification function takes only the values zero and one.

\begin{samepage}
  \begin{theorem}[Main result]\label{thm:main}
    There is a $0/1$ game $H$ and a projective reference strategy $\Sref$ on $\C^2\otimes\C^2$ with shared state $\ket{\Phi^+}=(\ket{00}+\ket{11})/\sqrt2$ such that:
    \begin{enumerate}
      \item The classical and quantum values satisfy
            \[
              \omega_{\mathrm c}(H)=\frac{71}{88} <\omega_{\mathrm q}(H)=\frac{60+11\sqrt2}{88}.
            \]
      \item Every optimal pure full-Schmidt-rank PVM strategy admits a local dilation to $\Sref$.
      \item There is an optimal pure full-Schmidt-rank strategy on $\C^2\otimes\C^2$ with one nonprojective measurement that admits no local dilation to $\Sref$.
    \end{enumerate}
  \end{theorem}

\end{samepage}

The construction has two steps. First, we build an auxiliary game $G$ with four Alice questions, three Bob questions, and classical and quantum value $19/22$. Every pure full-Schmidt-rank projective optimum returns a distinguished answer $0$ on every question. A second optimum uses $\ket{\Phi^+}$, three trine observables, and an additional three-outcome POVM, and never returns $0$. This gives two different optimal measurement behaviors, only one of which is available under the joint assumptions.

We then let the referee choose between $G$ and CHSH with equal probability, using disjoint question labels. The reference $\Sref$ uses $\ket{\Phi^+}$, standard CHSH measurements on the CHSH questions, and deterministic $0$ measurements on the $G$ questions. An optimal restricted strategy must maximize both components: CHSH supplies the state and its measurement relations, while $G$ fixes all the remaining measurement effects. Using the trine optimum on the $G$ questions gives a second optimum on the same Bell state, with the same CHSH measurements. Its different $G$ effects prevent a local dilation to $\Sref$. Thus every unrestricted optimum still certifies a Bell state, and the separation concerns certification of the complete strategy.

The main step in constructing $G$ is a bound for a three-setting Bell expression in terms of the total probabilities of giving nonzero answers. Both the deterministic behavior and the trine behavior saturate this bound. We choose question weights and penalized events so that an exact identity expresses the gap from the optimal score as a sum of nonnegative terms. Under full rank, its equality conditions become operator relations; imposing projectivity makes these relations incompatible with any nonzero nonabort subspace. The deterministic measurements that remain admit a one-dimensional reference strategy $\Szero$. The physical shared state is retained in the auxiliary system, as we make explicit in the proof.

The auxiliary game also permits a complete algebraic analysis. After compression to the local state supports, Alice's optimality relations have universal unital $C^*$-algebra
\[
  \mathcal B\cong\C\oplus M_2(\C).
\]
The scalar summand carries the deterministic behavior, and the matrix summand carries the trine behavior. The compatible local states are exactly the traces
\[
  \tau_\lambda(c\oplus a) =(1-\lambda)c+\frac{\lambda}{2}\Tr(a),\qquad 0\leq\lambda\leq1.
\]
They correspond affinely and bijectively to the finite-dimensional optimal correlations of $G$, with $\lambda$ equal to the nonabort probability. If $m_i$ are the three extra measurement effects, then
\[
  \mathcal B/\langle m_i^2-m_i:i=1,2,3\rangle\cong\C.
\]
Projectivity on the state support therefore leaves just the deterministic optimum. This is the algebraic mechanism used in the construction of $H$.

Projective realizations before support compression exhibit a different local state structure. For every optimal projective realization of $G$ with nonzero nonabort probability, Alice's zero left ideal
\[
  N_\varphi=\{a:\varphi(a^*a)=0\}
\]
is not two-sided; her local state is neither faithful nor tracial on the measurement algebra. A concrete realization generates $M_3(\C)$ and has a reduced state of rank two. Its compression gives the trine effects through a completely positive map with an explicit multiplicative defect. These calculations describe the distinction between state equations and operator relations in the example, complementing operator-algebraic formulations of self-testing such as~\cite{Paddock}.

Finally, the exact separation has a quantitative consequence. Perturbing the rank-two state of the projective dilation gives full-rank PVM strategies for $G$ whose scores tend to $19/22$ while their measurement dilation error stays at least one. Tensoring with a standard CHSH strategy gives the corresponding family for $H$ in local dimension six. Consequently, the self-test in Theorem~\ref{thm:main} is not robust.

The measurement ingredients have established precedents. The three-setting Bell expression and the three forbidden events occur in Pan~\cite[Eqs.~(10), (31)]{Pan}. The trine POVM and its three-dimensional dilation also occur in Baptista et al.~\cite[Section~6, Eqs.~(37)--(39)]{Baptista}. We use these ingredients to construct $G$ and then combine it with CHSH to obtain $H$. Section~\ref{sec:definitions} fixes the self-testing convention. Section~\ref{sec:game} constructs and analyzes $G$; Section~\ref{sec:chsh} completes the construction of $H$ and proves Theorem~\ref{thm:main}. Section~\ref{sec:algebra} gives the algebraic classification for $G$ and establishes the failure of robustness.

\section{Strategies and local dilation}\label{sec:definitions}

All Hilbert spaces used for tensor-product strategies are finite dimensional. We write $M_d(\C)$ for the algebra of complex $d\times d$ matrices. A positive operator-valued measurement (POVM) is a finite family of positive semidefinite operators summing to the identity. It is a projection-valued measurement (PVM) if every effect is a projection. Zero effects are allowed. A reflection is a self-adjoint operator whose square is the identity.

A two-player nonlocal game consists of finite question sets $\mathcal X$, $\mathcal Y$, finite answer sets $\mathcal O_A$, $\mathcal O_B$, a question distribution $\pi$, and a function $V(a,b\mid x,y)\in\{0,1\}$. A strategy $\mathsf S=(\rho,\{P_{x,a}\},\{Q_{y,b}\})$ has correlation
\[
  p(a,b\mid x,y)=\Tr\bigl(\rho(P_{x,a}\otimes Q_{y,b})\bigr)
\]
and winning probability
\[
  \omega_G(p)=\sum_{x,y,a,b}\pi(x,y)V(a,b\mid x,y)p(a,b\mid x,y).
\]
The quantum value $\omega_{\mathrm q}(G)$ is the supremum over all such strategies. We call a strategy optimal if it attains this value. A classical correlation is a convex combination of deterministic local response functions, and the classical value $\omega_{\mathrm c}(G)$ is the maximum score over these correlations. A correlation is nonlocal if it is not classical.

A pure state $\ket\psi\in H_A\otimes H_B$ has full Schmidt rank if both reduced density matrices are invertible. In particular, $\dim H_A=\dim H_B$. The elementary implication
\begin{equation}\label{eq:faithful}
  (T\otimes\Id)\ket\psi=0\quad\Longrightarrow\quad T=0
\end{equation}
then holds, and likewise on Bob's side. It follows directly by expanding $\ket\psi$ in a Schmidt basis with positive Schmidt coefficients.

We use the local-dilation convention of \cite[Definitions~2.5 and 2.7]{Baptista}. For clarity, we state also its approximate form, which is needed in Section~\ref{sec:robust}.

\begin{definition}\label{def:dilation}
  Let $\widetilde{\mathsf S}=(\ket{\widetilde\psi}, \{\widetilde P_{x,a}\},\{\widetilde Q_{y,b}\})$ be a pure reference strategy. It is a local $\varepsilon$-dilation of $\mathsf S$ if, for every purification $\ket\psi\in H_A\otimes H_B\otimes H_E$ of $\rho$, there are isometries
  \[
    U_A:H_A\longrightarrow\widetilde H_A\otimes K_A, \qquad U_B:H_B\longrightarrow\widetilde H_B\otimes K_B
  \]
  and a unit vector $\ket\eta\in K_A\otimes K_B\otimes H_E$ such that, writing $U=U_A\otimes U_B\otimes\Id_E$ and suppressing identity factors,
  \begin{align}
    \norm{U\ket\psi-\ket{\widetilde\psi}\otimes\ket\eta}&\leq\varepsilon,
    \label{eq:dilation-state}\\
    \norm{U(P_{x,a}\otimes\Id)\ket\psi- (\widetilde P_{x,a}\otimes\Id)\ket{\widetilde\psi}\otimes\ket\eta}
    &\leq\varepsilon,\label{eq:dilation-alice}\\
    \norm{U(\Id\otimes Q_{y,b})\ket\psi- (\Id\otimes\widetilde Q_{y,b})\ket{\widetilde\psi}\otimes\ket\eta}
    &\leq\varepsilon\label{eq:dilation-bob}
  \end{align}
  for all questions and answers. The tensor factors on the right are put in the corresponding order. A local dilation means $\varepsilon=0$. A game self-tests $\widetilde{\mathsf S}$ within a class of strategies if every optimal strategy in that class admits this local dilation. Unrestricted self-testing uses all strategies defined above.
\end{definition}

The measurement expressions in this definition use the effects themselves, not their square roots. Exact local dilation preserves correlations: take the inner product of the two single-party vectors in \eqref{eq:dilation-alice} and \eqref{eq:dilation-bob}.

\section{An auxiliary game}\label{sec:game}

We first construct a game $G$ whose optimal measurement behavior changes when full rank and projectivity are imposed simultaneously. It will supply the measurement component of $H$.

\subsection{Rules and optimal strategies}

Alice's questions are $1,2,3,4$ and Bob's are $1,2,3$. Their fixed answer sets are
\[
  \mathcal O_A=\{0,1,2,3\},\qquad \mathcal O_B=\{0,+1,-1\}.
\]
The answer $0$ means abort. On Alice's first three questions, answers $1,2$ are interpreted as signs $+1,-1$, respectively, and answer $3$ always loses. On question $4$, her answers $1,2,3$ are three distinct nonabort outcomes. Define
\begin{equation}\label{eq:weights}
  W=\begin{pmatrix}3&2&2\\2&2&2\\2&2&2\\1&1&1\end{pmatrix},
  \qquad \pi(x,y)=\frac{W_{xy}}{22}, \qquad C=J_3-2\Id_3,
\end{equation}
where $J_3$ is the all-ones matrix. The winning conditions are as follows.
\begin{enumerate}
  \item On $(1,1)$, both players must give nonabort signs and the signs must be opposite.
  \item On $(i,j)\ne(1,1)$ with $i\leq3$, both may abort; otherwise both must give signs $a,b\in\{+1,-1\}$ with $ab=C_{ij}$.
  \item On $(4,j)$, both may abort; otherwise Alice must answer $k\in\{1,2,3\}$ and Bob a sign $b$, with $(k,b)\ne(j,+1)$.
\end{enumerate}
Every outcome not specified as winning loses. In particular, exactly one player aborting always loses. These rules define a $0/1$ game, denoted by $G$.

For question sets with a distinguished answer $0$, define the \emph{deterministic abort reference} $\Szero$ by
\begin{equation}\label{eq:abort-reference}
  \widetilde H_A=\widetilde H_B=\C,\qquad \ket{\widetilde\psi}=\ket0\otimes\ket0,\qquad \widetilde P_{x,a}=\delta_{a,0}\Id_{\C},\qquad \widetilde Q_{y,b}=\delta_{b,0}\Id_{\C}.
\end{equation}
Its state has full Schmidt rank and its measurements are PVMs.

The reference $\Szero$ loses only on $(1,1)$, so its winning probability is $19/22$.

\begin{proposition}\label{prop:auxiliary}
  The game $G$ has the following properties.
  \begin{enumerate}
    \item $\omega_{\mathrm c}(G)=\omega_{\mathrm q}(G)=19/22$.
    \item Every optimal pure full-Schmidt-rank PVM strategy has abort effect $\Id$ on every question and all other effects zero. Consequently, $G$ self-tests $\Szero$ within this class.
    \item A pure full-Schmidt-rank strategy on $\C^2\otimes\C^2$, using a single nonprojective measurement, also attains $19/22$ and never aborts. Thus $G$ does not unrestrictedly self-test $\Szero$.
  \end{enumerate}
\end{proposition}

In particular, the restricted optima consist exactly of arbitrary pure full-Schmidt-rank states equipped with deterministic abort measurements. In a local dilation to $\Szero$, this physical state becomes the auxiliary state; the explicit maps are given in Lemma~\ref{lem:abort-dilation}.

\subsection{Construction from trine correlations}

The three qubit observables used in the construction are
\begin{equation}\label{eq:trine}
  Z=\begin{pmatrix}1&0\\0&-1\end{pmatrix},\qquad
  X=\begin{pmatrix}0&1\\1&0\end{pmatrix},\qquad
  T_1=Z,\quad T_2=-\frac Z2+\frac{\sqrt3 X}{2},\quad T_3=-\frac Z2-\frac{\sqrt3 X}{2}.
\end{equation}
They are reflections and sum to zero. The word \emph{trine} will refer only to these three equally spaced directions and their associated effects. The Bell expression $\sum_{ij}C_{ij}\langle A_i\otimes B_j\rangle$ reaches $6$ for $A_i=T_i$, $B_j=-T_j$ on the maximally entangled qubit state. The three-outcome POVM $M_k=(\Id+T_k)/3$ also satisfies $p(k,+1\mid4,k)=0$ on that state. These are the measurement components from the prior constructions discussed in the introduction.

Simply adding an abort outcome would give the abort strategy Bell score zero, rather than six. To make both behaviors optimal, replace the constant bound by a bound in terms of nonabort effects $E_i,F_j$:
\begin{equation}\label{eq:conceptual-bound}
  L:=\sum_{ij}C_{ij}\langle A_i\otimes B_j\rangle -\sum_i\langle E_i\rangle-\sum_j\langle F_j\rangle\leq0.
\end{equation}
The qubit strategy gives $6-3-3=0$, while the abort strategy gives zero term by term. Penalizing inconsistent abort outcomes and the three events $p(k,+1\mid4,k)$ preserves both equalities. Under full rank and projectivity, the vanishing penalties become incompatible operator relations on any nonzero nonabort subspace.

It remains to implement this idea with a verification function in $\{0,1\}$. In the qubit strategy, each of the six off-diagonal base question pairs loses with probability $1/4$. Weight two on each such pair therefore costs a total of three relative to the abort strategy. Assigning weight three to $(1,1)$ and accepting only opposite nonabort signs compensates exactly for this cost. Both strategies win on the other two diagonal pairs and on the three additional pairs. This explains the weights in \eqref{eq:weights}. The following identity establishes the common optimal value.

\subsection{Quantum value and equality conditions}\label{sec:certificate}

First assume that Alice's invalid base effects are zero. On base questions write
\[
  A_i=P_{i,+}-P_{i,-},\quad E_i=P_{i,+}+P_{i,-},\qquad B_j=Q_{j,+}-Q_{j,-},\quad F_j=Q_{j,+}+Q_{j,-}.
\]
The signs in Alice's subscripts stand for her answers $1,2$. For arbitrary POVMs,
\begin{equation}\label{eq:variance}
  A_i^2\preceq E_i,\qquad B_j^2\preceq F_j.
\end{equation}
For example, for effects $R_0,R_+,R_-$ define the isometry $Vv=(R_0^{1/2}v,R_+^{1/2}v,R_-^{1/2}v)$ and $D=\operatorname{diag}(0,\Id,-\Id)$. Then $A=V^*DV$, $E=V^*D^2V$, and $E-A^2=V^*D(\Id-VV^*)DV\succeq0$.

Use any purification of the shared state and put $a_i=(A_i\otimes\Id)\ket\psi$, $b_j=(\Id\otimes B_j)\ket\psi$, with the purifying identity suppressed. Since $C^2=4\Id_3-J_3$, direct expansion gives
\begin{align}
  -L={}&\sum_i\langle E_i-A_i^2\rangle +\sum_j\langle F_j-B_j^2\rangle\notag\\
  &+\sum_i\norm{a_i-\frac12\sum_j C_{ij}b_j}^{\!2} +\frac14\norm{\sum_jb_j}^{\!2}\geq0.
  \label{eq:sos}
\end{align}
Here and below brackets mean expectation in the shared state, with the appropriate identity tensor factors understood.

To relate $L$ to the winning probability, define
\[
  m^A_{ij}=p(a\ne0,b=0\mid i,j),\qquad m^B_{ij}=p(a=0,b\ne0\mid i,j).
\]
Let $s_{11}$ be the probability of two equal nonabort signs on $(1,1)$, and let $\ell_{4j}$ be the losing probability on $(4,j)$. Set
\begin{equation}\label{eq:penalty-matrices}
  K_A=\begin{pmatrix}1&1&1\\2&3/2&3/2\\2&3/2&3/2\end{pmatrix},
  \qquad
  K_B=\begin{pmatrix}1&2&2\\1&3/2&3/2\\1&3/2&3/2\end{pmatrix}.
\end{equation}
The exact score identity is
\begin{equation}\label{eq:score-identity}
  22\omega_G-19 =L-s_{11}-\sum_{ij}(K_A)_{ij}m^A_{ij} -\sum_{ij}(K_B)_{ij}m^B_{ij}-\sum_j\ell_{4j}.
\end{equation}

For completeness, this identity can be checked using the marginal coefficient arrays
\[
  \alpha=\begin{pmatrix}-1&1&1\\0&1/2&1/2\\0&1/2&1/2\end{pmatrix},
  \qquad
  \beta=\begin{pmatrix}-1&0&0\\1&1/2&1/2\\1&1/2&1/2\end{pmatrix}.
\]
Each row of $\alpha$ and each column of $\beta$ sums to one. Consequently, using the question-independent marginals of a quantum correlation, the coefficient of a base outcome in $L$ is $C_{ij}ab-\alpha_{ij}\mathbf1_{a\ne0} -\beta_{ij}\mathbf1_{b\ne0}$, where abort is represented by the number zero. At an ordinary pair $(i,j)\ne(1,1)$, we have $\alpha_{ij}+\beta_{ij}=1$. The comparison with minus twice the losing indicator is
\begin{center}
  \begin{tabular}{lccc}
    \toprule
    Outcome & coefficient in $L$ & $-2\mathbf1_{\rm loss}$ & difference\\
    \midrule
    Both abort & $0$ & $0$ & $0$\\
    Only Alice nonabort & $-\alpha_{ij}$ & $-2$ & $2-\alpha_{ij}$\\
    Only Bob nonabort & $-\beta_{ij}$ & $-2$ & $2-\beta_{ij}$\\
    Correct nonabort signs & $0$ & $0$ & $0$\\
    Incorrect nonabort signs & $-2$ & $-2$ & $0$\\
    \bottomrule
  \end{tabular}
\end{center}
At $(1,1)$ compare instead with three times the winning indicator. The difference is one on each single-abort event and each equal-sign nonabort event, and zero otherwise. The eight ordinary pairs have total baseline weight $16$, and the extra pairs have baseline weight $3$. These observations give \eqref{eq:score-identity}, including its constant.

Equations~\eqref{eq:sos} and \eqref{eq:score-identity} prove the upper bound $19/22$ when there are no invalid base answers. For a general strategy, merge Alice's invalid answer $3$ into her abort answer on questions $1,2,3$. This cannot turn a win into a loss, so it cannot lower the score. The modified measurements are still POVMs, and still PVMs if the original ones were PVMs. Thus the same upper bound holds for the fixed answer sets of $G$. Since $\Szero$ attains it, this proves Proposition~\ref{prop:auxiliary}\textup{(i)}, including the classical value since $\Szero$ is deterministic.

\subsection{Full-rank projective optima}

To determine when equality is possible for full-rank projective strategies, we use the following elementary obstruction.

\begin{lemma}\label{lem:obstruction}
  On a nonzero Hilbert space, there do not exist reflections $R_1,R_2,R_3$ and orthogonal projections $N_1,N_2,N_3$ such that
  \[
    R_1+R_2+R_3=0,\qquad \sum_iN_i=\Id,\qquad R_iN_i=N_i.
  \]
\end{lemma}
\begin{proof}
  For distinct $j,k$, both $N_jR_jN_k$ and $N_jR_kN_k$ vanish. The sum relation makes the third block $N_jR_iN_k$ vanish as well. Hence every $N_j$ commutes with every $R_i$. On a nonzero range of $N_i$, the other two reflections satisfy $R_j+R_k=-\Id$. Substitution of $R_k=-\Id-R_j$ into $R_k^2=\Id$ gives $\Id+2R_j=0$, contradicting $R_j^2=\Id$.
\end{proof}

Consider now an optimal pure full-Schmidt-rank PVM strategy without invalid base answers. Every nonnegative term in the certificate vanishes. Since the coefficients of $m^A_{ij},m^B_{ij}$ are strictly positive,
\begin{equation}\label{eq:flags}
  (E_i\otimes\Id)\ket\psi=(\Id\otimes F_j)\ket\psi \quad\text{for every }i,j.
\end{equation}
Indeed, the squared norm of the difference is the sum of the two abort-mismatch probabilities, since $E_i,F_j$ are projections. By \eqref{eq:faithful}, all $E_i$ equal one projection $P$, and all $F_j$ equal one projection $Q$. In particular $A_i^2=P$, $B_j^2=Q$. Writing $M_i=P_{4,i}$, the vanishing extra losses similarly gives
\begin{equation}\label{eq:extra-partition}
  \sum_{i=1}^3M_i=P,\qquad P_{4,0}=\Id-P.
\end{equation}

The vanishing squared norms in \eqref{eq:sos} yield $\sum_jb_j=0$ and $a_i=-b_i$, hence
\begin{equation}\label{eq:mirror}
  \sum_iA_i=0,\qquad (A_i\otimes\Id)\ket\psi=-(\Id\otimes B_i)\ket\psi.
\end{equation}
The forbidden event at $(4,i)$ has zero probability, so $(M_i\otimes Q_{i,+})\ket\psi=0$. Since $Q_{i,+}=(Q+B_i)/2$, equations~\eqref{eq:flags} and \eqref{eq:mirror} turn this into $(M_i(P-A_i)\otimes\Id)\ket\psi=0$. Full rank gives
\begin{equation}\label{eq:support-relations}
  M_iA_i=A_iM_i=M_i.
\end{equation}
On $PH_A$, the $A_i$ are reflections summing to zero and the $M_i$ form an orthogonal partition of the identity. Lemma~\ref{lem:obstruction} forces $P=0$. Equation~\eqref{eq:flags} and full rank then give $Q=0$. All measurements therefore abort.

For an optimal strategy with the original answer sets, perform the merging operation above. The merged strategy is optimal and thus always aborts, in particular on Bob's side. For every Alice base question $i$ there is an ordinary pair $(i,j)\ne(1,1)$ of positive weight. On this pair, an original invalid answer loses whereas the merged abort answer wins. Equality of the original and merged scores forces the invalid answer probability to be zero. Full rank makes its effect zero. Thus the original strategy also always aborts.

We have identified all the measurement effects of every restricted optimum. The following elementary observation supplies the local dilation.

\begin{lemma}\label{lem:abort-dilation}
  Every strategy with $P_{x,a}=\delta_{a,0}\Id$ and $Q_{y,b}=\delta_{b,0}\Id$ admits a local dilation to $\Szero$, regardless of its shared state.
\end{lemma}
\begin{proof}
  For any purification $\ket\psi\in H_A\otimes H_B\otimes H_E$, take $K_A=H_A$, $K_B=H_B$, and
  \begin{equation}\label{eq:abort-isometries}
    U_Av=\ket0\otimes v,\qquad U_Bw=\ket0\otimes w,\qquad \ket\eta=\ket\psi.
  \end{equation}
  After ordering the tensor factors as in Definition~\ref{def:dilation}, these maps give
  \begin{align*}
    U\ket\psi&=\ket{00}\otimes\ket\eta,\\
    U(P_{x,a}\otimes\Id)\ket\psi &=\delta_{a,0}\ket{00}\otimes\ket\eta =(\widetilde P_{x,a}\otimes\Id)\ket{00}\otimes\ket\eta,\\
    U(\Id\otimes Q_{y,b})\ket\psi &=\delta_{b,0}\ket{00}\otimes\ket\eta =(\Id\otimes\widetilde Q_{y,b})\ket{00}\otimes\ket\eta.
  \end{align*}
  These are the three exact local-dilation equations.
\end{proof}

The construction retains all physical systems, including their entanglement, in the auxiliary factor; for a mixed state, $\ket\eta$ is the chosen purification. Full Schmidt rank was used above to force the deterministic measurement effects; the extraction in Lemma~\ref{lem:abort-dilation} requires no rank assumption.

Applying Lemma~\ref{lem:abort-dilation} with the physical state as the auxiliary state completes the proof of Proposition~\ref{prop:auxiliary}\textup{(ii)}.

\subsection{A nonprojective optimum}

Define $\Strine$ using
\begin{equation}\label{eq:qubit-strategy}
  \ket{\Phi^+}=\frac{\ket{00}+\ket{11}}{\sqrt2},\qquad A_i=T_i,\qquad B_j=-T_j,\qquad M_k=\frac{\Id+T_k}{3}.
\end{equation}
All abort effects and Alice's invalid base effects are zero. The sign effects are $(\Id\pm T_i)/2$ on Alice's side and $(\Id\mp T_j)/2$ on Bob's side. Alice uses the $M_k$ on question~4. These sum to $\Id$, are positive, and have nonzero eigenvalue $2/3$.

For real matrices, the identity $\bra{\Phi^+}R\otimes S\ket{\Phi^+}=\Tr(RS^{\mathsf T})/2$ gives
\[
  \langle A_i\otimes B_j\rangle=
  \begin{cases}-1&i=j,\\1/2&i\ne j.\end{cases}
\]
Thus the diagonal base pairs win surely and each off-diagonal base pair wins with probability $3/4$. Each extra pair wins surely because
\[
  p(k,+1\mid4,k)=\frac12\Tr\left( \frac{\Id+T_k}{3}\frac{\Id-T_k}{2}\right)=0.
\]
The weighted score is $3+2\cdot2+6\cdot2\cdot(3/4)+3=19$. This proves optimality of $\Strine$.

Its correlation differs from that of $\Szero$, so exact local dilation is impossible. More directly, Alice's abort effect on question~4 is zero, whereas the reference effect is $\Id$. The left side of \eqref{eq:dilation-alice} is then one for every choice of isometries and auxiliary unit vector. This proves Proposition~\ref{prop:auxiliary}\textup{(iii)}.

\section{The entangled-reference game}\label{sec:chsh}

We now combine $G$ with CHSH to construct $H$ and prove Theorem~\ref{thm:main}.

\subsection{The game and its reference strategy}

The CHSH game asks $u,v\in\{0,1\}$ uniformly and accepts signs $a,b$ exactly when $ab=(-1)^{uv}$. Its quantum value is $\gamma=(2+\sqrt2)/4$, attained by $\ket{\Phi^+}$ with Alice observables $Z,X$ and Bob observables $(Z+X)/\sqrt2,(Z-X)/\sqrt2$. These state and measurement effects are certified by the exact CHSH self-test~\cite{McKague}.

For $H$, take question sets
\[
  \mathcal X_H=\{g_1,g_2,g_3,g_4,c_0,c_1\},\qquad \mathcal Y_H=\{g_1,g_2,g_3,c_0,c_1\},
\]
where $g$ and $c$ label the two components. The referee selects $G$ or CHSH with equal probability and then samples its questions in the usual way:
\begin{equation}\label{eq:h-distribution}
  \pi_H(g_i,g_j)=\frac{W_{ij}}{44},\qquad \pi_H(c_u,c_v)=\frac18.
\end{equation}
All other question pairs have probability zero; set their verification function to zero. Keep the answer sets $\{0,1,2,3\}$ for Alice and $\{0,+1,-1\}$ for Bob. On the $g$ questions use the rules of $G$. On the $c$ questions Alice's answers $1,2$ represent signs $+1,-1$, respectively; Bob answers a sign. Every other answer on these questions loses.

The reference $\Sref$ shares $\ket{\Phi^+}$ and uses the following measurements:
\[
  \begin{array}{c|c|c}
    \text{Component}&\text{Alice}&\text{Bob}\\ \hline
    G&\text{always answer }0&\text{always answer }0\\
    \text{CHSH}&Z,\ X&(Z+X)/\sqrt2,\ (Z-X)/\sqrt2
  \end{array}
\]
For an observable $R$, the two sign effects are $(\Id\pm R)/2$; all invalid-answer effects are zero. Thus $\Sref$ is a full-rank projective strategy on $\C^2\otimes\C^2$.

\subsection{Values and exact self-testing}

For every strategy, the score in $G$ is at most $19/22$. The CHSH component is bounded by $\gamma$: merging invalid answers into valid signs cannot decrease its score, and gives an ordinary binary-answer strategy. The reference $\Sref$ attains both bounds. Therefore
\begin{equation}\label{eq:h-value}
  \omega_{\mathrm q}(H) =\frac12\frac{19}{22}+\frac12\frac{2+\sqrt2}{4} =\frac{60+11\sqrt2}{88}.
\end{equation}
The classical component values are $19/22$ and $3/4$. Since the question sets are disjoint, deterministic response functions attaining them can be combined, giving
\[
  \omega_{\mathrm c}(H) =\frac12\frac{19}{22}+\frac12\frac34=\frac{71}{88}.
\]
This proves Theorem~\ref{thm:main}\textup{(i)}.

Every optimal strategy for $H$ must maximize both component scores. For a pure full-Schmidt-rank PVM strategy, Proposition~\ref{prop:auxiliary} consequently makes all its $G$ measurements deterministic. We next apply the CHSH self-test, accounting for the additional answers in the fixed alphabets.

Let $E_u,F_v$ be the total valid-answer effects on the CHSH questions, and let $A_u,B_v$ be the differences of the two valid sign effects. For any strategy,
\[
  \omega_{\rm CHSH} =\frac18\sum_{u,v}\langle E_u\otimes F_v +(-1)^{uv}A_u\otimes B_v\rangle \leq\frac{4+2\sqrt2}{8}.
\]
The CHSH bound $2\sqrt2$ applies to self-adjoint contractions, and each $\langle E_u\otimes F_v\rangle$ is at most one. Equality forces all four of these expectations, and hence their marginals, to be one. Every invalid effect therefore annihilates every purification of the shared state. Merging those effects into valid signs gives binary measurements with the same action on each purification. The exact CHSH self-test supplies local isometries extracting $\ket{\Phi^+}$ and the usual CHSH measurement relations.

For a full-rank strategy, the invalid effects are zero as operators. The resulting isometries also handle the deterministic $G$ effects: for every $g$ question the physical and reference answer-$0$ effects are identities, while all other effects are zero. The state extraction equation thus supplies the answer-$0$ measurement equation, and the remaining equations vanish. This proves Theorem~\ref{thm:main}\textup{(ii)}.

To prove part~\textup{(iii)}, use $\ket{\Phi^+}$ and the standard CHSH measurements, and use $\Strine$ on the $G$ questions. Both component scores are optimal. All measurements are projective except Alice's three-outcome trine POVM on $g_4$. Her answer-$0$ effect on $g_4$ is zero, whereas the corresponding reference effect is $\Id_2$. For any isometries and auxiliary unit vector, the left side of \eqref{eq:dilation-alice} for this effect is one. An exact local dilation to $\Sref$ is therefore impossible.

The two qubit optima use the same Bell state and the same CHSH measurements. More generally, the CHSH argument above extracts a Bell state from every unrestricted optimum of $H$. The distinction between the optimal strategies lies in their $G$ measurements.

\section{Algebraic consequences}\label{sec:algebra}

We now analyze the auxiliary game $G$, which supplies the measurement separation in $H$. On the local state support, optimality gives an algebra whose representations and compatible states can be classified completely. The deterministic abort reference corresponds to its scalar summand; the trine strategy corresponds to its matrix summand. We first establish this classification and then examine projective realizations on larger local spaces and the role of their local states.

\subsection{Supported relations and optimal states}\label{sec:classification}

We begin with the elementary representation-theoretic fact behind both the algebra and the optimal-state classification.

\begin{lemma}\label{lem:trine-algebra}
  Suppose $R_1,R_2,R_3$ are reflections on a nonzero finite-dimensional space with $\sum_iR_i=0$. There are a space $K$ and a unitary identifying $R_i$ with $T_i\otimes\Id_K$. If positive operators $N_i$ also satisfy $\sum_iN_i=\Id$ and $R_iN_i=N_i$, then necessarily
  \begin{equation}\label{eq:unique-povm}
    N_i=\frac{\Id+R_i}{3},\qquad N_i^2=\frac23N_i.
  \end{equation}
\end{lemma}
\begin{proof}
  Squaring $R_i+R_j=-R_k$ gives $R_iR_j+R_jR_i=-\Id$ for $i\ne j$. Consequently $Z'=R_1$, $X'=(R_2-R_3)/\sqrt3$ satisfy
  \[
    (Z')^2=(X')^2=\Id,\qquad Z'X'=-X'Z'.
  \]
  The reflection $X'$ maps the $+1$ eigenspace of $Z'$ unitarily onto its $-1$ eigenspace. Choosing this identification makes $Z'=Z\otimes\Id_K$, $X'=X\otimes\Id_K$, giving the stated form of the $R_i$.

  Put $q_i=(\Id+T_i)/2$, a rank-one projection. Positivity and $R_iN_i=N_i$ imply $N_i=q_i\otimes K_i$ for some $K_i\succeq0$. Comparing the $\Id,X,Z$ coefficients in $\sum_iq_i\otimes K_i=\Id_2\otimes\Id_K$ gives $K_1=K_2=K_3=2\Id_K/3$. This proves \eqref{eq:unique-povm}.
\end{proof}

The equality conditions of Section~\ref{sec:game} suggest the following relations for the nonabort projection and Alice's effects. We will show below that they hold for every optimal correlation after passing to the state supports. A universal unital $C^*$-algebra for a list of bounded operator relations means a unital $C^*$-algebra generated by elements satisfying those relations, such that every other realization of the relations induces a unique unital $*$-homomorphism from it. Consider the relations
\begin{equation}\label{eq:universal-relations}
  \begin{gathered}
    p=p^*=p^2,\qquad x_i=x_i^*,\qquad x_i^2=p,\qquad\sum_i x_i=0,\\
    m_i\succeq0,\qquad\sum_i m_i=p,\qquad x_im_i=m_i \quad(i=1,2,3).
  \end{gathered}
\end{equation}
Here $p$ represents Alice's common nonabort projection after compression to the local state support. All the generators have norm at most one.

\begin{proposition}\label{prop:algebra}
  The universal unital $C^*$-algebra $\mathcal B$ of \eqref{eq:universal-relations} is
  \begin{equation}\label{eq:algebra}
    \mathcal B\cong\C\oplus M_2(\C).
  \end{equation}
  Inside this algebra,
  \begin{equation}\label{eq:algebra-defect}
    m_i=\frac{p+x_i}{3},\qquad m_i^2=\frac23m_i, \qquad\sum_i(m_i-m_i^2)=\frac p3.
  \end{equation}
  Under the identification in \eqref{eq:algebra},
  \[
    p=0\oplus\Id_2,\qquad x_i=0\oplus T_i,\qquad m_i=0\oplus\frac{\Id_2+T_i}{3}.
  \]
  If $\mathcal J$ is the closed two-sided $*$-ideal generated by the $m_i^2-m_i$, then
  \begin{equation}\label{eq:quotient}
    \mathcal J=0\oplus M_2(\C),\qquad \mathcal B/\mathcal J\cong\C.
  \end{equation}
\end{proposition}
\begin{proof}
  The relations $x_i^2=p$ and $x_i=x_i^*$ imply $px_i=x_ip=x_i$. Similarly $0\preceq m_i\preceq p$ implies $pm_i=m_ip=m_i$. Thus $p$ is central. On the $(\Id-p)$ summand all listed generators vanish and only scalar multiples of the identity remain.

  On the $p$ summand, with unit $p$, put $z=x_1$ and $t=(x_2-x_3)/\sqrt3$. They satisfy $z^2=t^2=p$, $zt=-tz$, and are self-adjoint. Every word reduces to the linear span of $p,z,t,zt$. The Pauli representation maps these four elements to linearly independent matrices spanning $M_2(\C)$. Therefore their universal algebra is $M_2(\C)$. The argument of Lemma~\ref{lem:trine-algebra}, valid also with an arbitrary multiplicity Hilbert space, forces $m_i=(p+x_i)/3$ in every representation. Thus the $m_i$ introduce no additional generators. This proves the universal property and \eqref{eq:algebra}, as well as \eqref{eq:algebra-defect}. Every generator of $\mathcal J$ belongs to the matrix summand, whereas \eqref{eq:algebra-defect} gives $p\in\mathcal J$. Since $p$ is the identity of that summand, $\mathcal J=0\oplus M_2(\C)$. The scalar representation survives, giving \eqref{eq:quotient}.
\end{proof}

The algebra has two irreducible representations up to unitary equivalence: the scalar representation and the defining representation of $M_2(\C)$. Arbitrary representations are direct sums of these, with multiplicity. The next step is to determine which states on this algebra are compatible with the bipartite optimality conditions.

Write $T(\mathcal B)$ for the tracial state space of $\mathcal B$: its elements are positive linear functionals $\tau$ with $\tau(\Id)=1$ and $\tau(ab)=\tau(ba)$ for all $a,b\in\mathcal B$. Since a full matrix algebra has a unique normalized trace,
\begin{equation}\label{eq:trace-family}
  T(\mathcal B)=\{\tau_\lambda:0\leq\lambda\leq1\},\qquad \tau_\lambda(c\oplus a)=(1-\lambda)c+\frac{\lambda}{2}\Tr(a).
\end{equation}
The following theorem identifies these traces with the states induced by optimal strategies.

To state the correspondence explicitly, define effects inside $\mathcal B$ by
\begin{equation}\label{eq:algebra-effects}
  \begin{gathered}
    e_{i,0}=e_{4,0}=f_{j,0}=\Id-p,\qquad e_{i,3}=0,\\
    e_{i,1}=(p+x_i)/2,\qquad e_{i,2}=(p-x_i)/2,\\
    e_{4,k}=m_k,\qquad f_{j,+1}=(p-x_j)/2,\qquad f_{j,-1}=(p+x_j)/2.
  \end{gathered}
\end{equation}
Here $i,j,k\in\{1,2,3\}$, and all indices use the actual answer labels of $G$. The $f_{j,b}$ represent Bob's effects transferred to Alice's algebra by the equality relations; they are not required to commute with the $e_{x,a}$ inside $\mathcal B$.

Let $p_0,p_1$ be the correlations of $\Szero,\Strine$, respectively, and let $\Qq$ denote the set of all finite-dimensional quantum correlations with the question and answer sets of $G$.

\begin{theorem}[Optimal correlations and tracial states]\label{thm:face}
  For every $\tau\in T(\mathcal B)$, the formula
  \begin{equation}\label{eq:state-correlation}
    p_\tau(a,b\mid x,y)=\tau(e_{x,a}f_{y,b})
  \end{equation}
  defines an optimal quantum correlation. The map $\tau\mapsto p_\tau$ is an affine bijection from $T(\mathcal B)$ onto the optimal correlations of $G$ in finite dimensions. It sends $\tau_\lambda$ to $(1-\lambda)p_0+\lambda p_1$. Thus
  \begin{equation}\label{eq:face}
    F:=\{p\in\Qq:\omega_G(p)=19/22\} =\{(1-\lambda)p_0+\lambda p_1:0\leq\lambda\leq1\}.
  \end{equation}
  Every finite-dimensional optimal strategy, after purification and compression to the Schmidt supports, induces the corresponding state $\tau_\lambda$ on $\mathcal B$. For any base question $i$, the parameter is recovered by
  \[
    \lambda=\tau_\lambda(p)=\Pr(\text{Alice does not abort}\mid i).
  \]
\end{theorem}
\begin{proof}
  Any finite-dimensional quantum correlation has a pure full-Schmidt-rank POVM realization: purify the state, assigning the purifying space to one party for this purpose, and compress both parties' effects to their Schmidt supports. This operation preserves the correlation; it need not be a local dilation in the sense of Definition~\ref{def:dilation}. Write $\rho_A$ for Alice's reduced density matrix in this realization. Initially merge invalid base answers into abort, as before.

  For an optimal full-rank POVM realization without invalid answers, \eqref{eq:sos} gives $E_i=A_i^2$ and $F_j=B_j^2$ as operator identities: each difference is positive with zero expectation against an invertible marginal. Since $P_{i,\pm}=(A_i^2\pm A_i)/2\succeq0$ and $\norm{A_i}\leq1$, every eigenvalue $t$ of $A_i$ satisfies $t^2\pm t\geq0$ and $|t|\leq1$. Thus $\spec(A_i)\subseteq\{-1,0,1\}$. The base measurements are therefore PVMs, as are Bob's. The vanishing abort mismatches now give common nonabort projections $P,Q$ and \eqref{eq:flags}.

  Let $S=\sum_iM_i=\Id-P_{4,0}$, which at this stage is only an effect. For any fixed Bob question, zero abort mismatch gives
  \[
    \norm{(S\otimes\Id-\Id\otimes Q)\ket\psi}^{\!2} \leq\langle S\otimes(\Id-Q)+(\Id-S)\otimes Q\rangle=0.
  \]
  Together with \eqref{eq:flags} and full rank, this proves $S=P$. Hence $0\preceq M_i\preceq P$. The forbidden-event argument still applies to positive $M_i$: the positive operator $M_i\otimes Q_{i,+}$ has zero expectation and so annihilates $\ket\psi$. It follows that \eqref{eq:mirror} and \eqref{eq:support-relations} hold unchanged. Thus $p\mapsto P$, $x_i\mapsto A_i$, $m_i\mapsto M_i$ defines a representation $\pi_A$ of $\mathcal B$.

  All measurements are block diagonal for the abort and nonabort subspaces. Equation~\eqref{eq:flags} implies the orthogonal decomposition
  \[
    \ket\psi=\ket{\psi_0}+\ket{\psi_1},\qquad \ket{\psi_0}\in(\Id-P)H_A\otimes(\Id-Q)H_B,\quad \ket{\psi_1}\in PH_A\otimes QH_B.
  \]
  Put $\lambda=\norm{\psi_1}^2$. The abort component contributes $(1-\lambda)p_0$. If $\lambda>0$, the normalized active state has full Schmidt rank on $PH_A\otimes QH_B$. Lemma~\ref{lem:trine-algebra} identifies Alice's observables as $T_i\otimes\Id$ and her extra effects as $(\Id+T_i)\otimes\Id/3$.

  To determine the active state and Bob's observables, note that the mirror relation with a self-adjoint $B_i$ makes Alice's reduced density matrix $\rho_A^{(1)}$ commute with every $A_i$. This follows by taking the partial trace of the two sides of $(A_i\otimes\Id)\ket{\psi_1}\bra{\psi_1}$ and using \eqref{eq:mirror} also on the bra. Hence $\rho_A^{(1)}=\Id_2/2\otimes\sigma_A$ for an invertible density matrix $\sigma_A$ on $K_A$. A Schmidt decomposition, or uniqueness of purification, now identifies the active state by a Bob unitary with $\ket{\Phi^+}\otimes\ket\eta$. Full Schmidt rank and \eqref{eq:mirror} then force Bob's observables to be $-T_i^{\mathsf T}\otimes\Id=-T_i\otimes\Id$ on his active space. Thus the active correlation is exactly $p_1$. More specifically, pulling Alice's local state back through $\pi_A$ gives
  \[
    \Tr(\rho_A\,\pi_A(c\oplus a)) =(1-\lambda)c+\lambda\Tr\bigl((\Id_2/2\otimes\sigma_A) (a\otimes\Id)\bigr) =\tau_\lambda(c\oplus a).
  \]
  Terms whose weight is zero are omitted. This proves the asserted traciality on $\mathcal B$ without imposing any traciality assumption on the physical strategy. The block diagonal measurements do not detect cross terms between $\psi_0$ and $\psi_1$, so the correlation is $(1-\lambda)p_0+\lambda p_1$.

  For the original fixed answer sets, equality of scores before and after merging implies that invalid base answers have zero probability. Indeed, in the merged optimum an invalid answer, now counted as abort, can occur only with Bob aborting. On any ordinary pair in its row this changes an original loss into a merged win. Its marginal must therefore be zero, and full Schmidt rank makes its supported effect zero. Thus the original strategy induces the same representation $\pi_A$. This proves the asserted inclusion for all optimal correlations. The reverse inclusion follows by taking block direct sums of the two endpoint measurement families and superposing their normalized shared states with amplitudes $\sqrt{1-\lambda}$ and $\sqrt\lambda$. For these endpoints, \eqref{eq:state-correlation} follows directly from the abort effects and the maximally entangled qubit identity, respectively. Linearity proves it for every $\tau_\lambda$. These quantum realizations also establish positivity and normalization of the probabilities in that formula. The nonabort marginal recovers $\tau_\lambda(p)=\lambda$, so the correspondence is injective as well as surjective.
\end{proof}

The theorem identifies the role of the matrix summand in observed correlations: its weight is exactly the weight of the nonabort endpoint. It also distinguishes the universal algebra from the algebra generated in a particular supported realization. The latter is $\C$ at $\lambda=0$, $M_2(\C)$ at $\lambda=1$, and $\C\oplus M_2(\C)$ when both weights are positive. The pulled-back endpoint traces are not faithful on the universal algebra, but the local state is faithful on the algebra of its supported realization. The auxiliary density matrix $\sigma_A$ may be arbitrary; traciality on the measurement algebra does not require it to be maximally mixed.

Combining the theorem with Proposition~\ref{prop:algebra} explains the restricted self-test algebraically. Requiring the supported effects $m_i$ to be projections kills $\mathcal J=0\oplus M_2(\C)$. Of the compatible states $\tau_\lambda$, only $\tau_0$ vanishes on this ideal and hence factors through the scalar quotient. In a full-Schmidt-rank PVM strategy the support is the entire local space, so its measurement relations really do define a representation of this quotient. The nonabort weight must therefore be zero. This is why the restriction removes every optimal correlation except $p_0$.

\begin{corollary}\label{cor:classical}
  The classical value of $G$ is $19/22$, and its only optimal classical correlation is $p_0$. In particular, every point of $F\setminus\{p_0\}$ is nonlocal, even though $G$ has no quantum advantage in winning probability.
\end{corollary}
\begin{proof}
  Each deterministic strategy is a one-dimensional full-rank PVM strategy. By Proposition~\ref{prop:auxiliary}, its score can equal $19/22$ only when it always aborts. A classical mixture attains the bound only if every component with positive weight does so. Hence its correlation must be $p_0$.
\end{proof}

\subsection{Why projective realizations require state information}\label{sec:dilation}

The trine POVM has a three-dimensional projective dilation whose shared state has local rank two. We describe its measurement algebra and compression map, then characterize the local states of optimal projective realizations with nonzero nonabort probability.

Let $V:\C^2\to\C^3$ embed the first two coordinates and write $s=VV^*=\operatorname{diag}(1,1,0)$. Define the orthonormal vectors
\begin{equation}\label{eq:dilation-vectors}
  e_1=\frac1{\sqrt3}\begin{pmatrix}\sqrt2\\0\\1\end{pmatrix},\qquad
  e_2=\frac1{\sqrt6}\begin{pmatrix}-1\\-\sqrt3\\\sqrt2\end{pmatrix},\qquad
  e_3=\frac1{\sqrt6}\begin{pmatrix}-1\\\sqrt3\\\sqrt2\end{pmatrix},
  \qquad \Pi_i=e_ie_i^*.
\end{equation}
As in \cite[Eqs.~(37)--(39)]{Baptista}, up to labels and the party on which the measurement acts, these projections dilate the trine POVM:
\begin{equation}\label{eq:compression}
  \Phi:M_3(\C)\longrightarrow M_2(\C),\qquad \Phi(a)=V^*aV,\qquad \Phi(\Pi_i)=M_i.
\end{equation}
The dilation is minimal: the vectors in the ranges of the $\Pi_iV$ span $\C^3$.

On Alice's base questions take observables $T_i\oplus0$ and abort projection $\Id-s$. On question~4 take $\Pi_1,\Pi_2,\Pi_3$ and zero abort effect. With Bob's original qubit measurements and the state
\[
  \ket{\psi_*}=\frac{\ket{00}+\ket{11}}{\sqrt2} \quad\text{in }\C^3\otimes\C^2,
\]
this is a PVM realization of $p_1$. It is optimal, but Alice's reduced state is $s/2$ and hence does not have full rank. Notice that the base nonabort projection is $s$, whereas $\Pi_1+\Pi_2+\Pi_3=\Id_3$. Thus the supported relation $\sum_i m_i=p$ holds on the shared state and after compression, but fails as an operator identity in this realization.

\begin{proposition}\label{prop:compression}
  Alice's measurements in this realization generate $M_3(\C)$. The compression $\Phi$ is unital and completely positive, and
  \begin{equation}\label{eq:compression-defect}
    \Phi(\Pi_i^2)-\Phi(\Pi_i)^2 =\frac13M_i,\qquad \norm{\Phi(\Pi_i^2)-\Phi(\Pi_i)^2}=\frac29.
  \end{equation}
  The elements $a$ for which both
  \[
    \Phi(a^*a)=\Phi(a)^*\Phi(a),\qquad \Phi(aa^*)=\Phi(a)\Phi(a)^*
  \]
  hold are exactly the block diagonal algebra $M_2(\C)\oplus\C$ relative to $s$. In particular, compression is not a $*$-homomorphism.
\end{proposition}
\begin{proof}
  The base measurements generate $M_2(\C)\oplus\C$. Their algebra contains $s$ and all matrix units on the first two coordinates. Moreover, $s\Pi_1(\Id-s)=(\sqrt2/3)\ket0\bra2$ is nonzero. Multiplying this matrix unit by the base matrix units and taking adjoints produces every matrix unit in $M_3(\C)$.

  The map $\Phi$ preserves the identity, and its amplification to matrices of any size is again compression by an isometry; thus each amplification preserves positivity, which is the definition of complete positivity. For every $a\in M_3(\C)$,
  \begin{equation}\label{eq:general-defect}
    \Phi(a^*a)-\Phi(a)^*\Phi(a)=V^*a^*(\Id-s)aV.
  \end{equation}
  Its vanishing means $(\Id-s)as=0$. The analogous identity for $aa^*$ gives $sa(\Id-s)=0$. Together these say $as=sa$, proving the block diagonal characterization, often called the multiplicative domain of $\Phi$. Finally, $\Pi_i^2=\Pi_i$ and $M_i^2=2M_i/3$ give \eqref{eq:compression-defect}; $\norm{M_i}=2/3$ gives the norm.
\end{proof}

The change from $M_3(\C)$ to $M_2(\C)$ in this example is therefore compression by a completely positive map. It is not a quotient by a two-sided ideal: the simple algebra $M_3(\C)$ has no nonzero $*$-homomorphism into $M_2(\C)$. The defect in \eqref{eq:compression-defect} measures precisely how a projective measurement becomes nonprojective on the state support.

The failure of multiplicativity has a direct consequence for optimality relations. For the same realization, define the local state $\varphi(a)=\Tr((s/2)a)$ and its zero space
\begin{equation}\label{eq:null-ideal}
  N_\varphi=\{a\in M_3(\C):\varphi(a^*a)=0\} =\{a:as=0\}.
\end{equation}
It is a left ideal, since $as=0$ implies $(ba)s=0$ for every $b$. It is not a two-sided ideal. Indeed, with $r=s-\Id$,
\begin{equation}\label{eq:left-ideal-witness}
  (r\otimes\Id)\ket{\psi_*}=0,\qquad (r\Pi_1\otimes\Id)\ket{\psi_*}=-\frac13\ket{2,0}\ne0.
\end{equation}
Thus $r\in N_\varphi$ but $r\Pi_1\notin N_\varphi$. The state equation $r\ket{\psi_*}=0$ is therefore not preserved under right multiplication by measurement effects. As a two-sided ideal generator, the nonzero $r$ generates all of $M_3(\C)$, so imposing $r=0$ as an operator relation would give the zero quotient.

There is also a necessary condition that does not depend on these particular matrices. A state $\varphi$ on a unital $C^*$-algebra is \emph{faithful} if $\varphi(a^*a)=0$ implies $a=0$, and is \emph{tracial} if $\varphi(ab)=\varphi(ba)$ for all $a,b$.

\begin{proposition}\label{prop:nontracial}
  In any optimal PVM realization of $G$ whose correlation is not $p_0$, Alice's state on the algebra generated by her measurements is neither faithful nor tracial. More precisely, its zero left ideal $N_\varphi$ is not a two-sided ideal. This statement holds also for a commuting-operator realization, consisting of a shared unit vector and two measurement families on one, possibly infinite-dimensional, Hilbert space, with every Alice effect commuting with every Bob effect.
\end{proposition}
\begin{proof}
  The score certificate uses only positivity and commutation between the two parties, so it applies in this model as well. Represent the shared state by a unit vector, purifying if necessary. After merging invalid answers, its equality conditions give, with the tensor factors omitted in the commuting model,
  \[
    (E_i-F_j)\ket\psi=0,\qquad \left(\sum_iA_i\right)\ket\psi=0,\qquad (A_i+B_i)\ket\psi=0.
  \]
  The extra zero losses give $(S-E_i)\ket\psi=0$, where $S=\sum_kM_k$, and the forbidden events give $M_i(E_i-A_i)\ket\psi=0$.

  Suppose $N_\varphi$ were a two-sided ideal. It is norm closed, so the quotient of Alice's algebra by $N_\varphi$ is a unital $C^*$-algebra with nonzero identity. All Alice operators annihilating the state become zero there. The vector equations consequently become exactly the operator equations used in the PVM proof: a common projection $p$, $a_i^2=p$, $\sum_i a_i=0$, $\sum_i m_i=p$, and $a_im_i=m_i$, with the $m_i$ projections. Lemma~\ref{lem:obstruction}, applied in a faithful representation of this quotient, forces $p=0$. Alice therefore aborts with probability one. The flag equalities force Bob to abort as well, and the original invalid answers have zero probability by the score-preserving merging argument. This gives $p_0$, a contradiction.

  A faithful state has $N_\varphi=\{0\}$, a two-sided ideal. For a tracial state, if $a\in N_\varphi$ and $b$ is arbitrary, then
  \[
    \varphi((ab)^*(ab))=\varphi(abb^*a^*) \leq\norm b^2\varphi(aa^*) =\norm b^2\varphi(a^*a)=0.
  \]
  Thus its zero left ideal is also closed under right multiplication. Both possibilities have been excluded.
\end{proof}

In the concrete realization, the representation of $M_3(\C)$ is faithful, while its local state is neither faithful nor tracial: for the matrix units $E_{uv}=\ket u\bra v$,
\[
  \varphi(E_{02}E_{20})=\tfrac12,\qquad \varphi(E_{20}E_{02})=0.
\]
After compression, the state on $M_2(\C)$ is the normalized trace $\operatorname{tr}_2=\Tr/2$, and $\varphi=\operatorname{tr}_2\circ\Phi$. The nonmultiplicative map $\Phi$ is therefore precisely what relates the nontracial projective realization to the tracial supported description.

Thus the scalar quotient in \eqref{eq:quotient} describes projectivity on the state support. Projective realizations before compression are related to that description through their local states and compression maps, as the example and Proposition~\ref{prop:nontracial} make explicit. The certificate also gives the commuting-operator value $19/22$, since it remains valid for commuting measurements and $\Szero$ attains the bound.

\subsection{Exact optimality and approximation}\label{sec:robust}

The algebraic distinction also explains why optimizing a Bell functional under the two restrictions can produce the correct supremum while giving an incomplete description of its optimizers. The following density statement isolates that issue at the level of correlations.

\begin{corollary}\label{cor:density}
  Let $\Qf\subseteq\Qq$ be the correlations admitting a pure full-Schmidt-rank PVM realization. In the Euclidean topology of the finite array of probabilities,
  \begin{equation}\label{eq:density}
    \overline{\Qf}=\overline{\Qq},\qquad \Qf\cap F=\{p_0\},\qquad \overline{\Qf\cap F}\ne\overline{\Qf}\cap F=F.
  \end{equation}
\end{corollary}
\begin{proof}
  Purification and simultaneous local Naimark dilation give a pure PVM realization of any $p\in\Qq$. Pad the two local spaces to a common dimension and extend every measurement projectively to the unused space. In Schmidt bases, replace any zero Schmidt coefficients by positive coefficients tending to zero, and normalize the vector. With the measurements fixed, this gives pure full-Schmidt-rank PVM correlations converging to $p$. The first equality follows. Proposition~\ref{prop:auxiliary} gives the second, and Theorem~\ref{thm:face} gives the last assertion.
\end{proof}

In particular, the supremum of every linear Bell functional is the same over $\Qf$ and $\Qq$. Agreement of all these optimal values does not imply agreement of the strategies or correlations attaining one of them. Here $F$ denotes the finite-dimensional optimal set in \eqref{eq:face}. Thus taking the closure of a class of correlations and selecting the correlations that exactly attain the optimum need not commute.

The loss of faithfulness also prevents a uniform robustness bound. Robustness within a class means that for every $\varepsilon>0$ there is a $\delta>0$ such that every strategy in that class of score at least $\omega_{\mathrm q}(G)-\delta$ admits a local $\varepsilon$-dilation to the reference. The strict positivity of $\delta$ is part of this convention.

\begin{proposition}\label{prop:nonrobust}
  For every $0<t<1$ there is a pure full-Schmidt-rank PVM strategy on $\C^3\otimes\C^3$ with score
  \begin{equation}\label{eq:near-optimal}
    \omega_G(\mathsf S_t)=\frac{19-3t^2}{22}
  \end{equation}
  whose local-dilation error to $\Szero$ is at least one. Consequently, the restricted self-test of Proposition~\ref{prop:auxiliary} is not robust.
\end{proposition}
\begin{proof}
  Keep Alice's three-dimensional PVMs from Section~\ref{sec:dilation}, and give Bob the base observables $-T_j\oplus0$ with abort projection $\Id-s$. Use the state
  \begin{equation}\label{eq:state-family}
    \ket{\psi_t}=\sqrt{\frac{1-t^2}{2}}(\ket{00}+\ket{11}) +t\ket{22}.
  \end{equation}
  All three Schmidt coefficients are positive. The base measurement effects are block diagonal. Both the normalized qubit part and the $\ket{22}$ part have base weighted score $16$, so the base score is exactly $16$ for every $t$. On each extra question, the qubit part wins surely. On the $\ket{22}$ part Bob aborts and Alice never aborts, so that part loses. Cross terms vanish because Bob's measurement is block diagonal. The three extra questions thus contribute $3(1-t^2)$, proving \eqref{eq:near-optimal}.

  Alice's abort effect on question~4 is still zero. In \eqref{eq:dilation-alice}, its image must be compared with $(\Id\otimes\Id)\ket{\widetilde\psi}\otimes\ket\eta$, a unit vector for the reference $\Szero$. The error is exactly one, independently of the isometries and auxiliary vector.
\end{proof}

For the auxiliary strategies $\mathsf S_t$, the source of the failure can be seen directly from the reduced state
\[
  \rho_t=\operatorname{diag}\left(\frac{1-t^2}{2}, \frac{1-t^2}{2},t^2\right).
\]
For $r=\Id-s$, define the state-dependent norm $\norm a_{\rho_t}=\sqrt{\Tr(\rho_t a^*a)}$. Then
\begin{equation}\label{eq:state-norm}
  \norm r=1,\qquad \norm r_{\rho_t}=t\longrightarrow0.
\end{equation}
An invertible reduced state permits exact inference from a zero state-vector error to a zero operator, as in \eqref{eq:faithful}. To make that inference uniform for small errors requires a uniform spectral lower bound. Specifically, $\rho\succeq\mu\Id$ implies $\norm a\leq\mu^{-1/2}\norm a_\rho$, whereas the smallest eigenvalue of $\rho_t$ tends to zero. Thus the algebraic inference from a state equation to an operator equation is exact for each $t>0$ but has no uniform quantitative bound along this family.

\begin{corollary}\label{cor:h-nonrobust}
  For every $0<t<1$, there is a pure full-Schmidt-rank PVM strategy for $H$ on $\C^6\otimes\C^6$ with score
  \begin{equation}\label{eq:h-near-optimal}
    \omega_{\mathrm q}(H)-\frac{3t^2}{44}
  \end{equation}
  whose local-dilation error to $\Sref$ is at least one. Consequently, the restricted self-test in Theorem~\ref{thm:main} is not robust.
\end{corollary}
\begin{proof}
  On each side use $\C^3\otimes\C^2$, with shared state
  \[
    \ket{\Psi_t} =\ket{\psi_t}_{A_0B_0}\otimes\ket{\Phi^+}_{A_1B_1}.
  \]
  For the $G$ questions use the PVMs of Proposition~\ref{prop:nonrobust} on the first factor; for the CHSH questions use the standard PVMs on the second factor. All six Schmidt coefficients are positive. The component scores are $(19-3t^2)/22$ and $\gamma$, so their average is \eqref{eq:h-near-optimal}. Alice's answer-$0$ effect on $g_4$ is still zero, while the reference effect is $\Id_2$. Its measurement dilation error is therefore one for every choice of isometries and auxiliary unit vector.
\end{proof}

The construction of $H$ separates exact certification of the complete strategy while preserving Bell-state extraction at every optimum. A further question is whether the joint assumptions can change state certification itself: can a game self-test a full-rank projective entangled reference within the restricted class while admitting an unrestricted optimum from which the reference state cannot be extracted?

\end{document}